\documentclass[11pt]{article}
\usepackage{mathrsfs}
\usepackage{amsmath}
\usepackage{amsfonts}
\usepackage{amssymb, amsmath, cite}
\usepackage{color}

\usepackage{mathrsfs}
\usepackage{dsfont}
\usepackage{amsthm}
\usepackage{mathrsfs}
\usepackage{amsmath}
\usepackage{amsfonts}
\usepackage[colorlinks,linkcolor=black,anchorcolor=blue,citecolor=blue]{hyperref}
\usepackage{amssymb, amsmath, cite}

\newtheorem{theorem}{Theorem}[section]
\newtheorem{lemma}{Lemma}[section]

\newcommand\be{\begin{equation}}
\newcommand\ee{\end{equation}}
\newcommand\ber{\begin{eqnarray}}
\newcommand\bea{\begin{eqnarray}}
\newcommand\eea{\end{eqnarray}}
\newcommand\eer{\end{eqnarray}}
\newcommand\berr{\begin{eqnarray*}}
\newcommand\eerr{\end{eqnarray*}}
\newcommand\lm{\lambda}

\newcommand\vep{\varepsilon}\newcommand{\ii}{\mathrm{i}}
\newcommand{\dd}{\mathrm{d}}\newcommand\bfR{\mathbb{R}}
\newcommand\e{\mathrm{e}}\newcommand\pa{\partial}
{\setlength\arraycolsep{2pt}
\newcommand\nn{\nonumber}
\newcommand{\BPS}{{\mbox{\rm \tiny{BPS}}}}
\usepackage{graphicx}
\setlength{\baselineskip}{18pt}{\setlength\arraycolsep{2pt}

\begin{document}

\title{Solutions to the Minimization Problem Arising in a \\ Dark Monopole Model in Gauge Field Theory}
\author{Xiangqin Zhang\\School of Mathematics and Statistics\\ Henan University\\
 Kaifeng, Henan 475004, PR China\\\\
 Yisong Yang\\Courant Institute of Mathematical Sciences\\New York University\\New York, NY 10012, USA}
\date{}
\maketitle

\begin{abstract}
We prove the existence of dark monopole solutions in a recently formulated Yang--Mills--Higgs theory model
with technical features similar to the classical monopole problems. The solutions are obtained as energy-minimizing static spherically symmetric field configurations of unit topological charge.
We overcome the difficulty of recovering the full set of boundary conditions by a regularization method which
may be applied to other more complicated problems concerning monopoles and dyons in non-Abelian gauge field theories.
Furthermore we show in a critical coupling situation that an explicit BPS solution may be used to provide energy estimates for
non-BPS monopole solutions. Besides, in the limit of infinite Higgs coupling parameter, although no explicit construction is available, we establish an existence and uniqueness result for a
monopole solution and obtain its energy bounds.
\end{abstract}

\medskip
\begin{enumerate}

\item[]

\item[]

\item[]

\end{enumerate}

\section{Introduction}
\setcounter{equation}{0}

It is well known that the existence of a magnetic monopole was first theoretically conceived by P. Curie \cite{Curie} based on the
electromagnetic duality observed in the Maxwell equations. Later, Dirac \cite{Dirac} explored the quantum mechanical implication of the
presence of a magnetic monopole and demonstrated that the existence of such in nature would explain why electric charges are
quantized as multiples of a small common unit. Despite the elegance of Dirac's formalism, it is not surprising that, like a Coulomb electric point charge, a magnetic point charge, or simply
monopole, in the Maxwell theory would render a point singularity in the field and thus carry infinite amount of energy. 
 A landmark development came when 't Hooft \cite{tH} and Polyakov \cite{Pol} found in non-Abelian gauge field theory,
known as the Yang--Mills--Higgs theory, that a singularity-free topological soliton arises as a consequence of the
spontaneously broken symmetry in the vacuum manifold which behaves asymptotically like a Dirac monopole, away from local regions, has concentrated energy near the origin, where the monopole resides, is
smoothly distributed in space, and carries a finite energy. Such a particle-like soliton is commonly referred to as the 't Hooft--Polyakov
monopole and has since then been extensively studied. See \cite{GO,JT,Pre,Ra,Ry,Wein,Ybook,Zee} for some
expository presentations on the subject. Although monopoles remain a hypothetical construct, the concept they offer leads to
many fruitful investigations of various theoretical issues, for example, the quark confinement problem \cite{Gre,SY1,SY2} in light
of a linear confinement mechanism for a monopole and anti-monopole pair immersed in a type-II superconductor \cite{Man1,Man2,Nam,tH2,tH3}. In Deglmann and Kneipp \cite{DK}, a new family of monopole solitons are constructed with a characteristic feature that
they do not belong to the sector of the usual asymptotically electromagnetic Dirac monopole type, thus called `dark monopoles', which
may have some relevance in the study of dark matter/energy. As in the formulation of the 't Hooft--Polyakov monopole problem
\cite{Pol,tH}, mathematically, the existence of such monopoles amounts to solving a two-point boundary value problem with solutions
minimizing a correspondingly reduced radially symmetric energy functional as in
the classical work of Tyupkin, Fateev, and Shvarts \cite{Tyu} using functional analysis. However, the 
method in \cite{Tyu} is only sufficient to allow the acquisition of a weak solution
\cite{Burzlaff} since the functional is not coercive enough for us to recover the full set of boundary properties required for regularity which is not directly implied by finite-energy condition. Technically, one needs to impose boundary conditions, 
imposed both at the point where the monopole resides and at spatial infinity, to carry out a
minimization process. Without preservation of such a full set of boundary conditions
for field configurations over the associated admissible class, there is no ensurance for the attainability of
the energy minimum, thus even the existence of a weak solution as a critical point of the energy functional may become problematic. 
Therefore, it is imperative to tackle the issue of recovering the full set of boundary conditions in the minimization treatment.
Our strategy here is 
to use the regularization method developed in \cite{Y} for a simpler but similar Yang--Mills--Higgs monopole problem \cite{LW} 
with vanishing Higgs potential as in the classical Bogomol'nyi  \cite{Bo} and Prasad--Sommerfield \cite{PS} (BPS) limit \cite{JT,Ra,Sut}.
Minimization problems
of similar technical subtleties arise in other treatments of non-Abelian monopole problems  \cite{Actor,GO,KZ,Pre,Wein}.
Such general applicability and appeal motivate our present mathematical investigation.


Analytically, our regularization method consists of the following steps: We first solve the problem away from the singular point
which permits an approximation of the concerned boundary condition at the singularity. This approximation enables us to realize a monotonicity
property of a minimizing sequence such that the monotonicity property is preserved in the limit as well. Thus we are able to see that the boundary
limit at the singularity would exist as a consequence. We then argue that the limiting value must be the desired one otherwise it would falsify the 
finite-energy condition.  Furthermore, we will use a  BPS solution
of the same form as in the classical studies \cite{Bo,PS} to estimate the monopole energy (mass) away from the BPS phase.
Besides, we will obtain monopole solutions when the coupling parameters satisfy some specific conditions and demonstrate their
applications.

An outline of the content of the rest of the paper is as follows. In Section 2, we review the dark monopole model \cite{DK} briefly, introduce the associated minimization problem, and
state our results. In Section 3, we prove the existence of an energy-minimizing solution in the general setting and establish some
qualitative properties of the solution.  We then comment on a BPS critical phase.
In Section 4, we illustrate how to use the BPS solution to estimate the energy of a non-BPS monopole in the zero Higgs potential
situation. In Section 5, we study the limiting situation when the Higgs coupling constant is set to be infinite, which also appears
in the classical 't Hooft--Polyakov model and is of independent interest. This problem becomes
simpler since it is a single-equation problem, although no solution is explicitly known.
Due to the non-convexity of the energy functional, it is not immediate to see that the
solution is unique. Nevertheless, we are able to establish a uniqueness result and derive some qualitative properties of the solution.

\section{Dark monopole model and existence results}
\setcounter{equation}{0}

Following \cite{DK}, use $\phi$ to denote a scalar field in the adjoint representation of a non-Abelian gauge group $G$ such as
$SU(n)$ ($n\geq5$) and $W_\mu$ ($\mu=0,1,2,3$) a gauge field taking values in the Lie algebra of $G$. In spherically symmetric static limit, the
scalar field and spatial
components of the gauge field are represented in terms of spherical coordinates $(r,\theta,\varphi)$ by the expressions
\bea
\phi(r,\theta,\varphi)&=&vS+\alpha_0 f(r)\sum_{m}Y^*_{lm}(\theta,\varphi)Q_m,\label{x2.1}\\
W_i(r,\theta,\varphi)&=&\frac{(u(r)-1)}{er^2}\epsilon_{ijk} x^j M_k,\label{x2.2}
\eea
where $e,v>0$ are parameters, $S, Q_m, M_k$  appropriate generators of $G$, $u(r)$ and $f(r)$  real-valued profile functions, 
$Y^*_{lm}$  suitable spherical harmonics,
and $\alpha_0>0$ depends on $v$ and $l$
in a specific way. In terms of such representation, the gauge-covariant derivatives of $\phi$ and
the magnetic field induced from $W_i$ are
\bea
D_i\phi&=&\pa_i\phi+\ii e[W_i,\phi]=\alpha_0\left(\frac{f'}r(x^i Y^*_{lm})+fu(\pa_i Y^*_{lm})\right)Q_m,\\
B_i&=&\left(\frac{u'}{er}P^{ik}_T+\frac{u^2-1}{er^2} P_L^{ik}\right)M_k,\quad  P_L^{ik}=\frac{x^i x^k}{r^2},\quad P^{ik}_L+P^{ik}_T=\delta^{ik}.
\eea
With these, the total Yang--Mills--Higgs energy  for the dark monopole model reads \cite{DK}:
\bea\label{1111}
E(\phi,W_i)&=&\int_{\bfR^3}\left\{\frac12 \mbox{Tr}(B_i B_i)+\frac12\mbox{Tr}(D_i\phi D_i\phi)+\frac\lm4\left(\mbox{Tr}(\phi\phi)-v^2\right)^2\right\}\,\dd x\nn\\
&\equiv &\frac{4\pi v}{e}\,I(u,f),
\eea
where $\lm>0$ is the Higgs coupling parameter, and 
\be\label{x26}
I(u,f)=\int_{0}^{\infty}\left\{2(u')^{2}+\frac{(1-u^{2})^{2}}{r^{2}}+
\frac{1}{3|\lambda_{p}|^{2}}\left(r^{2}(f')^{2}+l(l+1)f^{2}u^{2}\right)
+\frac{\lambda}{9e^{2}|\lambda_{p}|^{4}} r^{2}(f^{2}-1)^{2}\right\}\,\dd r,
\ee
with $\lm_p$ a fundamental weight of $G$ and the updated rescaled radial variable $ev r\mapsto r$, which is denoted by $\xi$ in \cite{DK}, so that the associated Euler--Lagrange equations of \eqref{x26} are
\bea
u''&=&\frac{l(l+1)}{6|\lm_p|^2}f^2 u+\frac{u(u^2-1)}{r^2},\label{x27}\\
f''+\frac{2f'}r&=&l(l+1)\frac{fu^2}{r^2}+\frac{2\lm}{3e^2|\lm_p|^2}f(f^2-1),\label{x28}
\eea
subject to the boundary conditions
\be\label{1.2}
f(0)=0,\quad u(0)=1,\quad f(\infty)=1,\quad u(\infty)=0.
\ee

It is clear that all the conditions except the first one, $f(0)=0$, in \eqref{1.2}, are consequences of finite energy.
From \eqref{x2.1}, we see that the first one in \eqref{1.2}, i.e., $f(0)=0$, is required to ensure regularity of the Higgs scalar field
which is not directly imposed by the finiteness of \eqref{x26}. It is this feature that needs to be dealt with care as described 
earlier.

In order to simplify our notation, it will be convenient to use the substitution
\begin{eqnarray}\label{xx210}
\frac{1}{3|\lambda_{p}|^{2}}= \alpha, \quad l(l+1)= \beta,
\quad \frac{\lambda}{3e^{2}|\lm_p|^2} = \frac{\gamma}{2}.
\end{eqnarray}
Thus our existence results regarding dark monopole solitons governed by the boundary-value problem consisting of
 \eqref{x27}--\eqref{1.2}  may be stated as follows.

\begin{theorem}\label{th}
 Consider the differential equations \eqref{x27}--\eqref{x28} subject to the boundary conditions \eqref{1.2}
governing a pair of profile functions $u(r)$ and $f(r)$ describing the spherically symmetric  static Higgs field $\phi$ and gauge field $W_i$ represented by the ansatz stated in \eqref{x2.1}--\eqref{x2.2}.

\begin{enumerate}
\item[(i)] For any coupling and group parameters, there exists a finite-energy solution minimizing the rescaled energy \eqref{x26} which
enjoys the properties $0<u(r),f(r)<1$ for $r>0$ and that $u(r)$ and $1-f(r)$ vanish at infinity exponentially fast, and $1-u(r)$
and $f(r)$ vanish at $r=0$ like power functions, following
some sharp asymptotic estimates in both cases.

\item[(ii)] When $\gamma=0$, $\beta=2$, and $\alpha>0$, the equations are equivalent to a BPS set of first-order equations for solutions with
a finite energy. This system of the BPS equations has a unique solution which depends explicitly on the free parameter $\alpha$
which coincides with the classical BPS monopole solution. In other words, this is a BPS situation.

\item[(iii)] In the non-BPS situation when $\gamma=0$ and $\alpha,\beta>0$ are arbitrary, the equations have an energy-minimizing
solution such that both $u(r)$ and $f(r)$ are monotone functions. Furthermore, the energy of the BPS solution obtained in (ii)
may be used to get some energy estimates which become exact at the critical point $\beta=2$.

\item[(iv)] When $\gamma=\infty$, the reduced governing equation has an energy minimizing solution for any $\alpha,\beta>0$.
Although the solution is not known explicitly,  it is unique and fulfills specific pointwise bounds and its energy estimates 
can be obtained through some concrete computations.

\end{enumerate}
\end{theorem}

In the subsequent sections, we prove various parts of the theorem.\footnote{In fact (ii) is covered in the
 classical study of the $SU(2)$ Yang--Mills--Higgs equations and the results stated are well known. We will comment on this
link at the end of Section 3.
The inclusion of (ii) in the theorem mainly serves to put the other results in perspective. } In doing so, we develop our methods for minimization of
energy, construction of energy-minimizing solutions, realization of asymptotic behavior, and energy estimation. Moreover, we
will  present and comment on the mathematical details of the results stated in the theorem.

\section{Solutions to equations of motion by regularized minimization}\label{s3}
\setcounter{equation}{0}\setcounter{remark}{0}
\setcounter{theorem}{0}\setcounter{remark}{0}
\setcounter{lemma}{0}\setcounter{remark}{0}

In terms of the suppressed parameters given in \eqref{xx210}, the energy functional \eqref{x26} becomes
\begin{eqnarray}\label{111}\
I(u,f)=\int_{0}^{\infty} \bigg\{2(u')^{2}+\frac{(1-u^{2})^{2}}{r^{2}}+
\alpha[ r^{2}(f')^{2}+\beta f^{2}u^{2}]
+\frac{\alpha\gamma}{2}r^{2}(f^{2}-1)^{2}\bigg\}\,{\mathrm{d}}r.
\end{eqnarray}
The equations \eqref{x27}--\eqref{x28}, or the Euler--Lagrange equations associated with \eqref{111}, are
\begin{eqnarray}\label{112}
u''&=&\frac{\alpha\beta}{2}f^{2}u+\frac1{r^{2}}\, {u(u^{2}-1)} \\ \label{113}
f''&=& -\frac{2}{r}f'+\frac\beta{r^{2}}\,{fu^{2}}+ \gamma f(f^{2}-1),
\end{eqnarray}
subject to the boundary conditions stated in \eqref{1.2}.
For our purpose, we shall obtain solutions to (\ref{112})--(\ref{113}) subject to (\ref{1.2}) as an energy-minimizing configuration
of the functional (\ref{111}). To this end,
set
\begin{eqnarray}\label{114}
\eta_0=\inf \{I(u,f)|(u,f)\in X\},
\end{eqnarray}
where the admissible set $X$ is defined to be
\berr
X&=&\{(u,f)\,| \,I(u,f)<\infty, \mbox{ and the functions }  u,f \mbox{ are absolutely continuous on any}\\
&&\mbox{ compact subinterval of }
 (0, \infty) \mbox{
 and satisfy  \eqref{1.2}} \}.
\eerr

First, we note that the structure of the functional (\ref{111}) indicates that we may always modify $(u,f)$ in $X$  if necessary,
 to lower the energy, to achieve the property
\be
0\leq u\leq 1,\quad 0\leq f\leq1. 
\ee
This property will be observed  in our minimization study to follow.
 
Next, let $\{(u_{n},f_{n})\}$  be a minimizing sequence of (\ref{114}). Then, for any pair of numbers $0 < a < b < \infty$,
$\{(u_{n},f_{n})\}$ is a bounded sequence in
the Sobolev space $W^{1,2}(a,b)$.  By a diagonal subsequence argument, we obtain the existence of a pair $u,f \in W^{1,2}_{\mathrm{loc}}(0,\infty)$, so that $I(u,f) < \infty$ and, by choosing a suitable subsequence if necessary, we may assume
without loss of generality $u_{n} \rightarrow u$, $f_{n} \rightarrow f(n \rightarrow \infty)$ weakly in $W^{1,2}(a,b)$ and strongly in $C[a,b]$ for any $0<a<b<\infty$. In fact, to achieve this, we may proceed as follows: For
each $m=1,2,\dots$, consider the sequence$\{(u_{n},f_{n})\}$ on $(a,b)$ with $a=a_m=\frac1m$ and $b=b_m=1+m$ in the way that we
extract a subsequence, $\{(u_{1,n},f_{1,n})\}$, from $\{(u_{n},f_{n})\}$,  which is weakly convergent in $W^{1,2}(a_1,b_1)$;
similarly, we extract a subsequence, $\{(u_{2,n},f_{2,n})\}$, from $\{(u_{1,n},f_{1,n})\}$, which is
weakly convergent in $W^{1,2}(a_2,b_2)$; inductively, we extract a subsequence $\{(u_{m,n},f_{m,n})\}$,
from $\{(u_{m-1,n},f_{m-1,n})\}$, $m\geq 2$, which is weakly convergent in $W^{1,2}(a_m,b_m)$. Hence the diagonal
subsequence $\{(u_{n,n},f_{n,n})\}$ is weakly convergent in $W^{1,2}(a_m,b_m)$ for any $m\geq1$ which establishes that
$\{(u_{n,n},f_{n,n})\}$ is weakly convergent in $W^{1,2}(a,b)$ for any $0<a<b<\infty$ as asserted.
We are yet to show that $(u,f)$ lies in $ X$. 
In other words, we need to verify the boundary condition \eqref{1.2}. To  do so, we note that the tricky, and perhaps the
most unnatural or indirect, part of (\ref{1.2}) is $f(0)=0$, which will be detailed later, and other parts are rather straightforward to see \cite{Tyu}.
For example, assuming $I(u_n,f_n)\leq \eta_0+1$ ($\forall n$), we have the uniform estimate
\be
|f_{n}(r)-1|\leq\int_{r}^{\infty}|f'_{n}(\rho)|\,\dd\rho 
\leq
 \bigg(\int_{r}^{\infty}|\rho^{2}(f'_{n}(\rho))^{2}|\,\dd\rho\bigg)^{\frac{1}{2}}
\bigg(\int_{r}^{\infty}\frac{\dd\rho}{\rho^{2}}\bigg)^{\frac{1}{2}} \leq \left(\frac{\eta_0+1}\alpha\right)^{\frac12}r^{-\frac12},\label{3}
\ee
showing that $f_{n}(r)\rightarrow 1$ as $r\rightarrow\infty$ uniformly. As a consequence, $\{u_n\}$ is a bounded sequence
in $W^{1,2}(0,\infty)$. These properties readily establish $f(\infty)=1,u(0)=1, u(\infty)=0$.
So it now remains  to establish $f(0)=0$. To this end, we proceed as follows by a regularization approach \cite{Y}.

 \begin{lemma} \label{lemma 2}
 Let $\{(u_{n},f_{n})\}$ be a minimizing sequence of \eqref{114}. Then we can modify the sequence $\{f_{n}\}$ so that it solves
the boundary-value problem
 \be\label{117}
f_{n}''=-\frac{2}{r}f_{n}'+\frac\beta{r^{2}}{f_{n}u_{n}^{2}}+\gamma f_{n}(f_{n}^{2}-1),\quad \frac1n<r<\infty;\quad f_n\left(
\frac1n\right)=0,\quad f_n(\infty)=1.
\ee
\end{lemma}

\begin{proof}
 Let $\{(u_{n},f_{n})\}$  be a minimizing sequence of \eqref{114} satisfying $I(u_{n},f_{n}) \leq \eta_0 +1$ ($\forall n$), say.
Introduce the cut-off function
\be\label{119}
\xi_{n}(r)= 0, \quad r\leq\frac{1}{n};\quad
   \xi_{n}(r)= 1,\quad r\geq\frac{2}{n};\quad
 \xi_{n}(r)=   nr-1,\quad\frac{1}{n}< r<\frac{2}{n}.
\ee
Then, using $0\leq f_n\leq1$ and the Schwarz inequality, we have
\begin{eqnarray}\nonumber
\int_{0}^{\infty}r^{2}((\xi_{n}f_{n})')^{2}\,\dd r&=&
\int_{\frac{1}{n}}^{\frac{2}{n}}r^{2}\left(n^{2}f_{n}^{2}+(nr-1)^{2}(f'_{n})^{2}+2n(nr-1)f_{n}f'_{n}\right)\,\dd r
+\int_{\frac{2}{n}}^{\infty}r^{2}(f'_{n})^{2}\,\dd r \\ \nonumber
&\leq& \int_{\frac{1}{n}}^{\frac{2}{n}}\left( r^{2}n^{2}+r^{2}(f_{n}')^{2}+2nr^{2}|f'_{n}|\right)\,\dd r
+\int_{\frac{2}{n}}^{\infty}r^{2}(f'_{n})^{2}\,\dd r \\ \nonumber
&\leq&\frac{7}{3n}+2n\bigg(\int_{\frac{1}{n}}^{\frac{2}{n}}r^{2}\,\dd r\bigg)^{\frac{1}{2}}\bigg
(\int_{\frac{1}{n}}^{\frac{2}{n}}r^{2}(f_{n}')^{2}\,\dd r\bigg)^{\frac{1}{2}}
+\int_{\frac{1}{n}}^{\infty}r^{2}(f'_{n})^{2}\,\dd  r \\ \nonumber
&\leq&\frac{7}{3n}+2\left(\frac7{3n}\right)^{\frac12}\left(\int_{0}^{\infty}r^{2}(f'_{n})^{2}\,\dd r \right)^{\frac12}
+\int_{0}^{\infty}r^{2}(f'_{n})^{2}\,\dd r.
\end{eqnarray}
Other terms are easily controlled. Hence we obtain 
\begin{eqnarray}\label{1.11}
\eta_0=\lim_{n\rightarrow\infty}I(u_{n},f_{n}) = \lim_{n\rightarrow\infty}I(u_{n},\xi_{n}f_{n}).
\end{eqnarray}
In other words, $\{(u_{n},\xi_{n}f_{n})\}$ is also a minimizing sequence. That is,  we are allowed to assume that $f_{n}$ satisfies  the truncated condition $f_{n}(r)=0$ for $r\leq \frac{1}{n}$, which is seen to be regularized since the singularity of \eqref{111}
is at $r=0$.

Given $n=1,2\dots$, consider the problem
\begin{eqnarray}\label{121}
\eta_{n}=\min\{I(u_{n},f)|f\in F_{n}\},
\end{eqnarray}
where
\ber
F_n=&&\left\{f \big|  f \mbox{ is absolutely continuous on any compact subinterval of $(0,\infty)$,}\right.\nn\\
&&\left.\mbox{
 $f(r) = 0$  for  $r < \frac{1}{n}$, and $f(\infty) = 1$}\right\}.\nn
\eer

Let $\{f^m\}$ be a minimizing sequence of \eqref{121}. As before, we can assume that
$ 0 \leq f^m \leq 1$. Thus for any $\frac1n<c< \infty$, the sequence $\{f^m\}$ is bounded in $W^{1,2}\left(\frac1n,c\right)$. A diagonal subsequence argument shows that there is a subsequence, which we still denote by $\{f^m\}$,  and there is an element $f_{n}\in W_{\mbox{\small loc}}^{1,2}\left(\frac1n,\infty\right)$  (say), so that $f^m\rightarrow f_{n}$ $(m\rightarrow \infty)$
weakly in  $W^{1,2}\left(\frac1n, c\right)$ ($\forall c>\frac1n$). It is clear that $f_n$ solves (\ref{121}). So $f_n$ solves \eqref{117}
as well.
\end{proof}

\begin{lemma} \label{lemma 3} Let the pair $(u,f)$ be as obtained earlier as the limit of an appropriately
chosen minimizing sequence $\{(u_n,f_n)\}$ of the problem \eqref{114} constructed in Lemma \ref{lemma 2}. Then $f$ fulfills the desired boundary condition
\begin{eqnarray}\label{4}
\lim_{r\rightarrow 0} f(r)=0.
\end{eqnarray}
\end{lemma}
\begin{proof}
Assume that $f_{n}$ satisfies \eqref{117}. Taking $n\rightarrow \infty$ in any
interval $[a,b]$ with $0<a<b<\infty$ and iterating convergence from lower- to higher-order derivatives of the sequence, we see that 
the limit $f$ solves the equation
\begin{eqnarray}\label{122}
(r^{2}f')'= \beta fu^{2}+\gamma r^{2}f(f^{2}-1),\quad r>0.
\end{eqnarray}

We claim
\begin{eqnarray}\label{123}
\lim_{r\rightarrow 0}\inf r^{2}|f'(r)|= 0.
\end{eqnarray}
In fact, if \eqref{123} is false, there are constants $\delta>0$ and $\varepsilon_{0}>0$
such that
\begin{eqnarray}\label{x116}
r^{2}|f'(r)|>\varepsilon_{0},\quad 0<r<\delta.
\end{eqnarray}
Therefore, in view of \eqref{x116}, for any $0<r_0<\delta$, we have
\begin{eqnarray}
\int_{r_{0}}^{\delta}r^{2}(f'(r))^{2}\,\dd r>
\varepsilon_{0}^{2}\left(\frac{1}{r_{0}}-\frac{1}{\delta}\right),
\end{eqnarray}
which diverges as $r_{0}\rightarrow 0$,  contradicting the convergence of the integral $\int_{0}^{\infty}r^{2}(f'(r))^{2}\,\dd r$.

Using \eqref{123}, we obtain after integrating \eqref{122} that
\begin{eqnarray}\label{124}
r^{2}f'(r)=\int_{0}^{r}\beta f(\rho)u^{2}(\rho)\,\dd\rho
+\int_{0}^{r} \gamma \rho^{2}f(\rho)(f^{2}(\rho)-1)\,\dd\rho.
\end{eqnarray}
Using $0\leq f\leq1$ and $u(0)=1$ in \eqref{124}, we see that
 that $f'(r)\geq0 $ when $r>0$ is small. In particular, the monotonicity of $f(r)$ holds for $r>0$ small. Consequently, there is a number $f_{0}\geq 0$ such that 
\be\label{x119}
\lim_{r\rightarrow 0}f(r)=f_{0}.
\ee
Moreover, \eqref{124} implies
\be\label{x120}
\lim_{r\to0} r^2 f'(r)=0.
\ee
With such preparation, we are now ready to prove $f_0=0$ in \eqref{x119}.

In fact,  in view of (\ref{x120}), we can apply the L'H\^{o}pital's rule to deduce the result
\begin{eqnarray} 
\lim_{r\rightarrow 0}rf'(r)=\lim_{r\rightarrow 0}\frac{r^{2}f'(r)}{r}=\lim_{r\to0}\left(r^2 f'(r)\right)'
=\beta f_{0},
\end{eqnarray}
where we have inserted \eqref{122} to get the right-hand-side quantity of the above. 
Hence,  if $f_{0}> 0$ in \eqref{x119}, then there are constants $\delta>0$ and $\vep_{0}>0$ such that
\begin{eqnarray}\label{131}
   rf'(r)\geq \vep_{0}, \quad 0<r<\delta.
\end{eqnarray}
Integrating \eqref{131}, we obtain
\begin{eqnarray}
 f(r_{2})- f(r_{1})\geq\vep_{0}\ln\frac{r_{2}}{r_{1}},\quad r_{1},r_{2}\in(0,\delta),
\end{eqnarray}
which contradicts the existence of limit stated in \eqref{x119}.
Therefore the lemma follows. 
\end{proof}

In conclusion, we have $(u,f)\in X$ which solves \eqref{114}. Thus  $(u,f)$
is a least-energy solution of \eqref{112}--\eqref{113} subject to the boundary condition \eqref{1.2}.

We now present some properties of the energy-minimizing solution obtained above. 

\begin{lemma} \label{101}
The least-energy solution $(u,f)$ to \eqref{112}--\eqref{113} subject to the boundary condition \eqref{1.2}
enjoys the properties $0<u(r)<1, 0<f(r)<1$, for any $r>0$, and
\bea\label{x123}
u(r)&=&\mbox{\rm O}\left(\e^{-\sqrt{\frac{\alpha\beta}2}(1-\varepsilon)r}\right),\,
f(r)=1+\mbox{\rm O}\left(\e^{-\sqrt{2}\min\{\sqrt{\gamma},\sqrt{\alpha\beta}\}(1-\varepsilon)r}\right) (\gamma>0),\, r\rightarrow\infty,\\
u(r)&=&1+\mbox{\rm O} (r^{2(1-\vep)}),\quad f(r)=\mbox{\rm O}\left(r^{\left(\sqrt{\frac14+\beta}-\frac12\right)(1-\vep)}\right),\quad r\to0,\label{323}
\eea
 where  $\varepsilon\in(0,1)$ may be taken to be arbitrarily small. When $\gamma=0$, the estimate for $f(r)$ in \eqref{x123} is adjusted to
\be\label{ff}
f(r)=1+\mbox{\rm O}\left(\frac1r\right),\quad r\to\infty.
\ee
\end{lemma}

\begin{proof}
 Let $(u,f)$ be the energy-minimizing solution obtained.
Then
 $0\leq u(r)\leq 1, 0\leq f(r)\leq 1, r>0$.
Since $u=0$ and $f=0$ are equilibria of \eqref{112} and \eqref{113}, respectively, so they are not attainable at finite $r$
in view of the uniqueness theorem for the initial value problem of an ordinary differential equation. In other words,
$0<u(r)\leq 1$ and $ 0<f(r)\leq 1$, $r>0$. If there is a point $r_0>0$ such that $u(r_0)=1$, then $u'(r_0)=0$ and $u''(r_0)\leq0$.
Inserting these into \eqref{112}, we arrive at a contradiction since $f(r_0)>0$. So $u(r)<1$ for all $r>0$. Similarly, $f(r)<1$ for all
$r>0$ as well.

Furthermore, the asymptotic estimates stated in \eqref{x123}--\eqref{323} may be seen from analyzing the equations \eqref{112} and 
\eqref{113} rewritten in forms
\be\label{324}
u''=a(r) u,\quad f''+\frac2r f'=b(r) (f-1)+\frac{\beta}{r^2}fu^2,\quad r>0,
\ee
 for large $r>0$,
where the coefficients $a(r)$ and $b(r)$ satisfy the properties
\be\label{325}
\lim_{r\to\infty} a(r)=\frac{\alpha\beta}2,\quad \lim_{r\to\infty} b(r)=2\gamma,
\ee
in view of the behavior $u(\infty)=0, f(\infty)=1$, and 
\bea
U''&=&\frac{A(r)}{r^2} U+B(r),\label{326}\\
f''+\frac2r f'&=&\frac{\left(\beta u^2(r)+\gamma r^2(f^2-1)\right) }{r^2} f,\label{327}
\eea
for small $r>0$, where $U=u-1$ and the coefficients $A(r)$ and $B(r)$ satisfy
\be\label{328}
\lim_{r\to0} A(r)=2,\quad \lim_{r\to 0} B(r)=0,
\ee
in view of $u(0)=1$ and $f(0)=0$, respectively.   We study \eqref{323} first.

Note that the structure of \eqref{327} for $r>0$ small leads us to its linearized form around $f=0$:
\be\label{329}
F''+\frac2r F'=\frac{\beta}{r^2}F,
\ee
whose solution vanishing at $r=0$ is
\be
F(r)=C r^{\sigma},\quad \sigma=\sqrt{\frac14+\beta}-\frac12,\quad r>0,
\ee
where $\sigma$ is the positive root of the characteristic equation $\sigma^2+\sigma-\beta=0$ of
\eqref{329}. This suggests that we may choose the comparison function
\be\label{331}
F_\vep(r)=Cr^{\sigma_\vep},\quad r>0,\quad \sigma_\vep=\sigma(1-\vep),\quad \vep\in(0,1),
\ee
which gives us
\be\label{332}
F_\vep''+\frac2r F_\vep'=\frac{(\sigma^2_\vep+\sigma_\vep)}{r^2}F_\vep.
\ee
Now for any $\vep\in(0,1)$ there is $r_\vep>0$ such that 
\be\label{333}
\beta u^2(r)+\gamma r^2(f^2-1)>\sigma^2_\vep+\sigma_\vep,\quad r\in(0,r_\vep),
\ee
since $\sigma^2_\vep+\sigma_\vep<\beta$ and the left-hand side of \eqref{333} tends to $\beta$ as $r\to0$. Inserting \eqref{333}
into \eqref{327} and using \eqref{332}, we arrive at
\be\label{334}
(f-F_\vep)''+\frac2r(f-F_\vep)'>\frac{(\sigma_\vep^2+\sigma_\vep)}{r^2}(f-F_\vep),\quad r\in(0,r_\vep).
\ee
Let $C$ in \eqref{331} be a large enough positive constant such that $f(r_\vep)\leq F_\vep(r_\vep)$. In view of this and
$f(0)=F_\vep(0)=0$ as boundary condition and applying the maximum principle \cite{GT} to the differential inequality \eqref{334},
we have $f(r)-F_\vep(r)\leq 0$ for $r\in(0,r_\vep)$, which establishes the bound
\be
 0<f(r)\leq Cr^{\sigma(1-\vep)},\quad r\in(0,r_\vep),
\ee
resulting in the estimate for $f$ in \eqref{323}.

Now consider \eqref{326}. Take
\be\label{336}
F_\vep(r)=Cr^{2(1-\vep)},\quad r>0,\quad \vep\in\left(0,1\right).
\ee
Then
\be\label{337}
F_\vep''=\frac{2(1-\vep)(1-2\vep)}{r^2} F_\vep.
\ee
Combining \eqref{326} with \eqref{337}, we have
\be\label{338}
(U+F_\vep)''=\frac{A(r)}{r^2}(U+F_\vep)+\frac1{r^2}\left(2(1-\vep)(1-2\vep)-A(r)\right)F_\vep +B(r).
\ee
For fixed $\vep$ and any $C\geq1$ given in \eqref{336}, we see that when $r_\vep>0$ is sufficiently small, we have 
\be\label{339}
\frac1{r^2}\left(2(1-\vep)(1-2\vep)-A(r)\right)F_\vep(r) +B(r)\leq 0,\quad r\in(0,r_\vep),
\ee
in view of \eqref{328}. Inserting \eqref{339} into \eqref{338}, we obtain 
\be\label{340}
(U+F_\vep)''\leq\frac{A(r)}{r^2}(U+F_\vep),\quad r\in(0,r_\vep).
\ee
Choose $C$ in \eqref{336} sufficiently large such that $U(r_\vep)+F_\vep(r_\vep)\geq0$. In view of this, $U(0)+F_\vep(0)=0$,
the differential equality \eqref{340}, and the maximum principle \cite{GT}, we have $U(r)+F_\vep(r)\geq0$ for $r\in (0,r_\vep)$.
That is,
\be
-Cr^{2(1-\vep)}\leq u(r)-1<0,\quad r\in(0,r_\vep),
\ee
which establishes the estimate for $u(r)$ near $r=0$ stated in \eqref{323}.

For the asymptotic estimates \eqref{x123}, we note that the one for $u(r)$ is easy because of the form of the equation of $u$
in \eqref{324} in view of \eqref{325}. To get the estimate for $f(r)$ in \eqref{x123}, we use the comparison function
\be\label{342}
F(r)=C\e^{-\sigma r},\quad r>0,\quad C\geq1.
\ee
Then we have
\be
F''+\frac2r F'=\left(\sigma^2-\frac{2\sigma}r\right)F,\quad r>0,
\ee
which may be used in the equation of $f$ in \eqref{324} to give us
\be
(F+f-1)''+\frac2r(F+f-1)'=b(r)(F+f-1)+\left(\sigma^2-\frac{2\sigma}r-b(r)\right)F+\frac\beta{r^2} fu^2.
\ee
We now set
\be
\sigma=\sqrt{2}\min\{\sqrt{\gamma},\sqrt{\alpha\beta}\}(1-\vep),
\ee
for $\gamma>0$. Then, using \eqref{325} and the estimate for $u(r)$ stated in \eqref{x123}, we wee that there is some large $r_\vep>1$ such that
\be
b(r)>1 \quad\mbox{(say)},\quad \left(\sigma^2-\frac{2\sigma}r-b(r)\right)F(r)+\frac\beta{r^2} f(r)u^2(r)\leq0,\quad r\geq r_\vep.
\ee
Choose $C$ in \eqref{342} large enough such that $F(r_\vep)+f(r_\vep)-1\geq0$. Using this and that $F+f-1$ vanishes at infinity in
the differential inequality
\be
(F+f-1)''+\frac2r(F+f-1)'\leq b(r)(F+f-1),\quad r\geq r_\vep,
\ee
we obtain $F(r)+f(r)-1\geq0$ for all $r\geq r_\vep$ by the maximum principle \cite{GT}. This establishes
\be
-C\e^{-\sigma r}\leq f(r)-1<0,\quad r\geq r_\vep,
\ee
which gives rise to the asymptotic estimate for $f(r)$ stated in \eqref{x123} when $\gamma>0$.

The estimate for $f(r)$ when $\gamma=0$ stated in \eqref{ff} will be established in the next section.
\end{proof}

Thus the proof of part (i) of Theorem \ref{th} is carried out.

We now turn to part (ii) of Theorem \ref{th}.

 Thus, in \eqref{111}, consider the special situation, $\beta=2, \gamma= 0$, with the radial energy functional
\begin{eqnarray}\label{301}\
I(u,f)=\int_{0}^{\infty} \bigg\{2(u')^{2}+\frac{(1-u^{2})^{2}}{r^{2}}+
\alpha r^{2}(f')^{2}+2\alpha f^{2}u^{2}\bigg\}{\mathrm{d}}r,
\end{eqnarray}
which becomes the classical Bogomol'nyi \cite{Bo} and Prasad--Sommerfield \cite{PS} limit of the $SU(2)$ Yang--Mills--Higgs
monopole model with a vanishing Higgs coupling constant or zero Higgs potential density function, known as the BPS self-dual limit,
with setting $f\mapsto \frac f{\sqrt{\alpha}}$.  Thus, below, we only recall some facts which are useful our study,
although, for the purpose of our presentation, we keep the parameter $\alpha$ here
in order to relate the problem to the issues of our interest.

First, note that the occurrence of spontaneously broken symmetry dictates the asymptotic condition
\be\label{x22}
f(\infty)=f_\infty>0,
\ee
where $f_\infty$ is otherwise prescribed which is sometimes referred to as the monopole mass \cite{JT}. Due to the structure of (\ref{301}), it is seen that the energy is symmetric under the change of variables and parameter:
\be\label{x23}
u\mapsto u,\quad f\mapsto{f_\infty}f,\quad \alpha\mapsto \frac\alpha {f_\infty^2}.
\ee
Therefore we may assume $f_\infty=1$ in \eqref{x22} without loss of generality. We will observe this `normalized' asymptotic condition in the sequel. That is, we again follow the boundary condition (\ref{1.2}) for our problem.
The Euler--Lagrange equations associated with \eqref{301} are
\begin{eqnarray}
u''&=&\alpha f^{2}u+\frac{u(u^{2}-1)}{r^{2}},\label{302} \\ 
f''&=& -\frac{2f'}{r}+\frac2{r^{2}}\,{fu^{2}},\label{303}
\end{eqnarray}
which may also be obtained by setting $\beta=2$ and $\gamma=0$ in \eqref{112}--\eqref{113} and whose least-energy
solution may be obtained by minimizing \eqref{301} subject to \eqref{1.2} as before. 

Next, as in \cite{Bo,PS},
and using the boundary conditions \eqref{1.2}, we have
\begin{eqnarray}\label{304}\nonumber
I(u,f)&=&\int_{0}^{\infty} \bigg\{2\left(u'+\sqrt{\alpha}fu\right)^{2}+\left(\sqrt{\alpha}rf'(r)-\frac{(1-u^{2})}{r}\right)^{2}+
2\sqrt{\alpha}(f(1-u^{2}))'\bigg\}{\mathrm{d}}r\\
&\geq& 2\sqrt{\alpha}.
\end{eqnarray}
Hence, we have the energy lower bound, $I(u,f)\geq2\sqrt{\alpha}$,
which is attained when the pair $(u, f)$ satisfies the following BPS-type equations
\begin{eqnarray}\label{305}
&& u'+\sqrt{\alpha}fu=0,  \\ \label{306}
&& \sqrt{\alpha}rf'(r)-\frac{(1-u^{2})}{r}=0,
\end{eqnarray}
which may be solved \cite{PS} to yield the unique solution given
 explicitly by the formulas
\begin{eqnarray}\label{317}
 u(r)=\frac{\sqrt{\alpha}r}{\sinh \sqrt{\alpha}\,r},\quad
f(r)=\coth{\sqrt{\alpha}}r-\frac{1}{\sqrt{\alpha}\, r},\quad r>0.
\end{eqnarray}

It is interesting to note that, among finite-energy solutions satisfying the boundary condition
\eqref{1.2},
 the Euler--Lagrange equations \eqref{302}--\eqref{303} and  the BPS equations \eqref{305}--\eqref{306} are equivalent,
as established by Maison \cite{Maison}. In contrast, for the $SU(3)$ situation, Burzlaff \cite{Burzlaff,Burzlaff2} showed the existence of a non-BPS solution  even within radially symmetric configurations; in the $SU(2)$ setting, without radial symmetry assumption, Taubes
\cite{Taubes} established the existence of an infinite family of non-BPS solutions in the BPS
coupling, whose result was later extended in several important contexts \cite{Bor,Parker,SS,SSU,ST}. Thus, in general
non-Abelian gauge field theories, the equivalence
statement may not be valid.

\section{The general non-BPS situation with $\gamma=0$}\label{s1}
\setcounter{equation}{0}\setcounter{remark}{0}
\setcounter{theorem}{0}\setcounter{remark}{0}
\setcounter{lemma}{0}\setcounter{remark}{0}

In this section, we establish part (iii) of Theorem \ref{th}.
Thus we consider the energy \eqref{111}  when $\gamma=0$ such that the energy functional assumes the form
\begin{eqnarray}\label{041}
I(u,f)= \int_{0}^{\infty} \left\{2(u')^{2}+\frac{(1-u^{2})^{2}}{r^{2}}+
\alpha r^{2}(f')^{2}+\alpha\beta f^{2}u^{2}\right\}\,{\mathrm{d}}r.
\end{eqnarray}
As noted in the previous section, the arbitrary asymptotic limit given in \eqref{x22} may be normalized to fit  into that
stated in \eqref{1.2} through the
rescaling of parameters set in \eqref{x23}.
The Euler--Lagrange equations associated with \eqref{041} are
\begin{eqnarray}
u''&=&\frac{\alpha\beta}{2}f^{2}u+\frac{u(u^{2}-1)}{r^{2}},\label{043} \\ 
(r^2f')'&=&\beta{fu^{2}}.\label{044}
\end{eqnarray}
As before, it is readily shown that \eqref{043}--\eqref{044} has a solution $(u,f)$ that minimizes the energy \eqref{041}, satisfying
the boundary condition \eqref{1.2} and enjoying the property $0<u(r),f(r)<1$ for $r>0$. Thus, using \eqref{124} with $\gamma=0$, we see that
$f'(r)>0$ for all $r>0$. It is less obvious to see that $u(r)$ is also monotone as we now show below.

\begin{lemma} \label{lemma 2.3} Let $(u,f)$ be the solution pair to the equations \eqref{043}--\eqref{044} described above. 
Then the function $u(r)$ strictly decreases.
\end{lemma}
\begin{proof}
 We first
show that $ u$ is nonincreasing. Suppose otherwise that there are $0<a<b<\infty$ so that
$u(a) < u(b)$. Let $r_{1} \in(0, b)$ satisfy
\begin{eqnarray}
 r_{1}=\sup\left\{\hat{r}\in(0,b)\bigg|u(\hat{r})=\inf_{r\in (0,b)} u(r)\right\}.
\end{eqnarray}
Therefore $u(r) > u(r_{1})$ for all $ r \in (r_{1}, b)$. Since $u(r)\rightarrow 0$  as $r \rightarrow \infty$, we have a unique
$r_{2}>r_{1}$
satisfying
\begin{eqnarray}
 r_{2}=\inf\left\{r>r_{1}|u(r)= u(r_{1})\right\}.
\end{eqnarray}
(In fact, $r_2>b$.)
Now modify $u$ by setting
\begin{equation}
\tilde{u}(r)=
    u(r_{1}), \quad  r\in(r_{1},r_{2});\quad
   \tilde{u}(r)= u(r), \quad  r\notin(r_{1},r_{2}).
\end{equation}
Then $(\tilde{u},f)\in X$ but $I(\tilde{u},f)<I(u,f)$ which is false. In fact we have only to
compare the energies over the interval $(r_{1},r_{2})$. That is, we are to show that
$
\tilde{J}<J
$
where
\begin{eqnarray}
J&=& \int_{r_{1}}^{r_{2}}\left\{\left(2(u')^{2}+\frac{(1-u^{2})^{2}}{r^{2}}\right)+
\alpha r^{2}(f')^{2}+\alpha\beta f^{2}u^{2}
\right\}{\mathrm{d}}r,\\ 
\tilde{J}&=& \int_{r_{1}}^{r_{2}}\left\{\frac{(1-\tilde{u}^{2})^{2}}{r^{2}}+
\alpha r^{2}(f')^{2}+\alpha\beta f^{2}\tilde{u}^{2}
\right\}{\mathrm{d}}r.
\end{eqnarray}
We recall by the definition of $r_{1}$ that $u'(r_{1})=0$ and $u''(r_{1})\geq 0$
since $r_{1}$ is a local minimum point.
Inserting this information into  \eqref{043},  we find
\begin{eqnarray} \label{125}
\frac{\alpha\beta}{2} f^{2}(r_{1})+\frac{(u^{2}(r_{1})-1)}{r_{1}^{2}}\geq 0,
\end{eqnarray}
since $u(r_1)>0$.
On the other hand,
\begin{eqnarray}
J-\tilde{J}&=& \int_{r_{1}}^{r_{2}}\bigg\{2(u')^{2}+\left(\frac{(1-u^{2})^{2}}{r^{2}}-\frac{(1-\tilde{u}^{2})^{2}}{r^{2}}\right)+
\alpha \beta f^{2}(u^{2}-\tilde{u}^{2})\bigg\} {\mathrm{d}}r\nn\\ 
&\geq&\int_{r_{1}}^{r_{2}}\bigg\{2(u')^{2}+\bigg(u^{2}(r)-u^{2}(r_{1})\bigg)
\bigg(\alpha\beta f^{2}(r)+\frac{1}{r^{2}}\left(u^{2}(r)+u^{2}(r_{1})-2\right)\bigg)\bigg\}{\mathrm{d}}r\nn \\
&>&\int_{r_{1}}^{r_{2}}2(u')^{2}{\mathrm{d}}r+
\bigg({\alpha\beta} f^{2}(r_{1})+\frac{2(u^{2}(r_{1})-1)}{r_{1}^{2}}\bigg)\int_{r_{1}}^{r_{2}}\bigg(u^{2}(r)-u^{2}(r_{1})\bigg){\mathrm{d}}r\nn
 \\ 
&>& 0,
\end{eqnarray}
in view of  \eqref{125} and the fact that $f(r)$ increases. Consequently $u(r)$ can only be nonincreasing.

If $u$ is not strictly decreasing, it must be a constant in an interval. So we arrive at a
contradiction by  using the equation  \eqref{043} because it implies that $r^2 f(r)$ is constant, which
is false since $f(r)$ increases. Thus the lemma follows.
\end{proof}

We now estimate the energy carried by $(u,f)$. For convenience, we denote the energy \eqref{301} by $I(u,f;\alpha,\beta)$. Then we have by \eqref{304} the lower estimate
\be\label{311}
I(u,f;\alpha,\beta)\geq I(u,f;\alpha,2)\geq 2\sqrt{\alpha},\quad \beta\geq2.
\ee
Moreover, using \eqref{304} again, we have
\begin{eqnarray} \nonumber
I(u,f;\alpha,\beta)
&\geq& \int_{0}^{\infty} \bigg\{2(u')^{2}+\frac{(1-u^{2})^{2}}{r^{2}}+
\frac{\alpha\beta}{2}r^{2}(f')^{2}+\alpha\beta f^{2}u^{2}
\bigg\} {\mathrm{d}}r\\ 
&=&I\left(u,f;\frac{\alpha\beta}2,2\right)\nn\\
&\geq& \sqrt{2\alpha\beta},\quad 0<\beta<2.\label{312}
\end{eqnarray}
Summarizing \eqref{311} and \eqref{312}, we have the energy lower bound
\begin{eqnarray} \label{313}
I(u,f)
\geq \min\left\{\sqrt{2\alpha\beta},2\sqrt{\alpha}\right\},\quad\forall\alpha,\beta>0.
\end{eqnarray}

To get some upper estimates for the energy, we make the decomposition
\be\label{xx314}
I(u,f;\alpha,\beta)=I(u,f;\alpha,2)+\alpha(\beta-2)\int_0^\infty f^{2}u^{2}\,\dd r.
\ee
Now use the BPS solution \eqref{317}, denoted as $(u_\BPS,f_\BPS)$, as a test configuration to get
\begin{eqnarray}\label{xx315}
\int_0^\infty f_\BPS^{2}(r)u_\BPS^{2}(r)\,\dd r&=&
\int_0^\infty\left(\frac{\sqrt{\alpha}r}{\sinh \sqrt{\alpha}\,r}\right)^2
\left(\coth{\sqrt{\alpha}}r-\frac{1}{\sqrt{\alpha}\, r}\right)^2\,\dd r\nn\\
&=&\frac1{\sqrt{\alpha}}\int_0^\infty\left(\frac{r}{\sinh r}\right)^2
\left(\coth r-\frac{1}{ r}\right)^2\,\dd r\nn\\
&=&\frac1{3\sqrt{\alpha}}\left(\frac{\pi^2}6-1\right).
\end{eqnarray}
Thus, inserting $(u_\BPS,f_\BPS)$ into the right-hand side of \eqref{xx314} and using \eqref{xx315}, we have
\begin{eqnarray}
I(u,f;\alpha,\beta)&\leq& I(u_\BPS,f_\BPS;\alpha,\beta)\nn\\
&=&I(u_\BPS,f_\BPS;\alpha,2)+\alpha(\beta-2)\int_0^\infty f^2_\BPS u^2_\BPS\,\dd r\nn\\
&=&2\sqrt{\alpha}+\frac1{3}\left(\frac{\pi^2}6-1\right)\sqrt{\alpha}(\beta-2)\nn\\
&=&\frac{\sqrt{\alpha}}3\left(8-\frac{\pi^2}3+\left[\frac{\pi^2}6-1\right]\beta\right).\label{316}
\end{eqnarray}

Summarizing (\ref{313}) and (\ref{316}), we obtain the estimates of the energy of the solution pair $(u,f)$ as follows
\begin{eqnarray}\label{xx317}
\sqrt{\alpha} \min \bigg\{\sqrt{2\beta},2\bigg\}\leq I(u,f)
 \leq \frac{\sqrt{\alpha}}3\left(8-\frac{\pi^2}3+\left[\frac{\pi^2}6-1\right]\beta\right).
\end{eqnarray}
It is interesting that when $\beta=2$ we arrive at the BPS situation, $I(u,f)=2\sqrt{\alpha}$, as anticipated, which is hardly
surprising. Note also that, except for the critical situation $\beta=2$ where \eqref{xx317} becomes equality, in any
non-BPS situation $\beta\neq2$, inequalities in \eqref{xx317} are strict because the pair $(u_\BPS,f_\BPS)$ does not satisfy the coupled 
equations \eqref{043}--\eqref{044}.

In Figure \ref{F} we plot the energy lower and upper bounds given in \eqref{xx317} for $\frac{I(u,f)}{\sqrt{\alpha}}$ as functions of $\beta$.
\begin{figure}[h]
\begin{center}
\includegraphics[height=6cm,width=8cm]{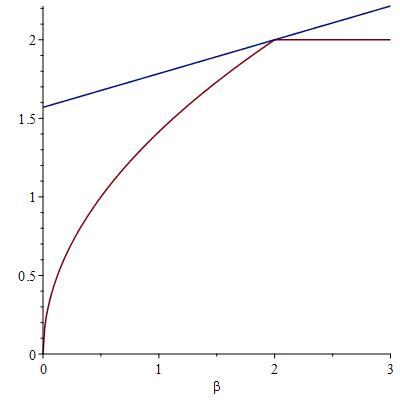}
\caption{A plot of the energy lower and upper bounds \eqref{xx317} for $\frac{I(u,f)}{\sqrt{\alpha}}$ as curves over $\beta>0$. It is seen that the
bounds coincide at  and stay close near the critical BPS point $\beta=2$.}
\label{F}
\end{center}
\end{figure}

We now verify the asymptotic estimate for $f(r)$ stated in \eqref{ff}. For this purpose, we integrate \eqref{044} and use \eqref{x120} to get
\be\label{418}
r^2 f'(r)=\int_0^r \beta f(\rho)u^2(\rho)\,\dd\rho.
\ee
In view of the estimate for $u(r)$ stated in \eqref{x123}, we see that the right-hand side is a bounded quantity for $r>0$. Thus,
integrating \eqref{418}, we have
\be
1-f(r)=\int_r^\infty \frac1{\eta^2}\left(\int_0^\eta  \beta f(\rho)u^2(\rho)\,\dd\rho\right)\,\dd\eta=\mbox{\rm O}\left(\frac1r\right), \quad r\to\infty,
\ee
which establishes \eqref{ff}. 

We note that the explicit BPS solution \eqref{317} confirms the asymptotic estimates stated in Lemma
\ref{101}
in the case when $\gamma=0$.

\section{The situation when $\gamma=\infty$}
\setcounter{equation}{0}

We now turn to part (iv) of Theorem \ref{th}.
In this situation, following \cite{DK} and with our notation, the energy functional now reads 
\be\label{41}
I(u)=\int_0^\infty \left\{ 2(u')^2+\frac{(1-u^2)^2}{r^2} +\alpha\beta u^2\right\}\,\dd r,
\ee
with the associated Euler--Lagrange equation subject to the corresponding boundary condition:
\be\label{42}
u''=\frac12\alpha\beta u+\frac1{r^2}(u^2-1)u,\quad r>0;\quad u(0)=1,\quad u(\infty)=0.
\ee

This type of problems also occur in other situations in gauge field theory (e.g., a discussion in the next section). Due to such 
separate interest,
we summarize
our existence and uniqueness results regarding \eqref{42} as follows.

\begin{theorem}
The boundary value problem \eqref{42} has a solution which minimizes the energy \eqref{41}. Furthermore, such a solution
satisfies the properties that $0<u(r)<1$ for $r>0$, $u(r)$ strictly decreases, and
\be\label{es}
1< \frac{I(u)}{\sqrt{2\alpha\beta}}< \sqrt{1+4\ln2}.
\ee
In fact, any finite-energy solution of \eqref{42}
enjoys the additional properties
\be\label{x63}
\lim_{r\to0} u'(r)=0,\quad u(r), u'(r)=\mbox{\rm O} \left(\e^{-\sqrt{\frac{\alpha\beta}2}(1-\vep)r}\right)\mbox{ as }r\to\infty,
\ee
where $\vep\in(0,1)$ may be taken to be arbitrarily small.
Besides, any nonnegative solution $u$ to \eqref{42} satisfies the global pointwise lower bounds
\be\label{x64}
\e^{-\sqrt{\frac{\alpha\beta}2}r}< u(r)<1,\quad r>0,
\ee
and is unique. In particular, subject to the boundary conditions in \eqref{42}, the energy \eqref{41} has a unique minimizer.
\end{theorem}

\begin{proof}
As before, it is not hard to prove that \eqref{42} has a solution which minimizes the energy \eqref{41}
and decreases monotonically. The energy estimates \eqref{es} will be obtained later.

In fact, the result about the limit of $u'(r)$ as $r\to0$ in \eqref{x63} holds as a consequence of the estimate \eqref{323}.  Furthermore, near $r=\infty$, the differential equation in $u$ in \eqref{42} may be
approximated by the linear equation $\eta''=\frac{\alpha\beta}2\eta$ which leads to the exponential decay estimates stated
in \eqref{x63} as well. Below we elaborate on \eqref{x64} and the uniqueness of a nonnegative solution in detail.

Let $u\geq0$ be a nonnegative solution of \eqref{42}. Then $u(r)>0$ for all $r>0$ otherwise there is some $r_0>0$ such that
$u(r_0)=0, u'(r_0)=0$, resulting in $u(r)=0$ for all $r>0$ by the uniqueness of a solution to the initial value problem of an
ordinary differential equation, which is false.
Moreover, using the maximum principle in \eqref{42}, we have $u<1$.
Hence we have $u''<\frac12\alpha\beta u$. Now let $\eta$ denote the left-hand-side exponential function in \eqref{x64}.
Then $\eta''=\frac12\alpha\beta \eta$ and $(\eta-u)''> \frac12\alpha\beta (\eta-u)$. Thus, using the boundary condition
that $\eta-u$ vanishes at $r=0$ and $r=\infty$, we get $(\eta-u)(r)<0$ for all $r>0$ in view of the maximum principle. So
\eqref{x64} follows.

Let $u_1$ and $u_2$ be two finite-energy nonnegative solutions of \eqref{42}. So $u_1(r),u_2(r)>0$ for all $r>0$. Set
$w=u_1-u_2$. Then $w$ satisfies
\begin{eqnarray}\label{0216}
w''=\frac{\alpha\beta}{2}\, w+\frac1{r^2}\left(u^2_1+u_1 u_2 +u_2^2-1\right)w.
\end{eqnarray}
Using  \eqref{42} with $u=u_1$ and \eqref{0216}, we get
\be \label{0217}
\frac{w}{u_{1}}(w'u_{1}-wu_{1}')'
=\frac1{r^2}\,({u_{1}u_{2}+u_{2}^{2}})\,w^{2}.
\ee
Integrating \eqref{0217} over the interval $0<r<R$, we have
\begin{eqnarray}\label{0218}
\frac{w}{u_{1}}(w'u_{1}-wu_{1}')\bigg|_{0}^{R}-\int_{0}^{R}\frac{(w'u_{1}-wu_{1}')^{2}}{u_{1}^{2}}{\mathrm{d}}r
=\int_{0}^{R}\frac1{r^2}({u_{1}u_{2}+u_{2}^{2}})\,w^{2}{\mathrm{d}}r.
\end{eqnarray}
On the other hand, in view of \eqref{x63}--\eqref{x64}, we have
\be\label{x68}
\lim_{r\to\infty} u_2 (r)\frac{u_1'(r)}{u_1(r)}=0.
\ee
Letting $R\to \infty$ in \eqref{0218} and applying \eqref{x63} and  \eqref{x68}, we arrive at
\begin{eqnarray}\label{0222}
-\int_{0}^{\infty}\frac{(w'u_{1}-wu_{1}')^{2}}{u_{1}^{2}}{\mathrm{d}}r
=\int_{0}^{\infty}\frac1{r^2}({u_{1}u_{2}+u_{2}^{2}})\,w^{2}{\mathrm{d}}r.
\end{eqnarray}
Since $u_1,u_2>0$, so $w\equiv0$ and the uniqueness result follows.
\end{proof}

{\em Note.} Since the energy functional \eqref{41} is not convex, the uniqueness of a critical point of it is generally not ensured. Our
theorem however asserts that  \eqref{42} has a unique minimizer as a solution to \eqref{42}.
\medskip

We now estimate the energy the unique minimizer of \eqref{41} carries.

First, let $u$ be a finite-energy solution of \eqref{42}. As a critical point of \eqref{41}, we see that the rescaled function 
$u_\delta(r)=u(\delta r)$ satisfies
\be\label{43}
\left(\frac{\dd  I(u_\delta)}{\dd\delta}\right)_{\delta=1}=0.
\ee
Inserting \eqref{41} into \eqref{43}, we arrive at the energy partition identity
\be
\int_0^\infty \left\{ 2(u')^2+\frac{(1-u^2)^2}{r^2} \right\}\,\dd r=\int_0^\infty \alpha\beta u^2\,\dd r,
\ee
resulting in the much simplified expression for the minimum energy
\be
I(u)= 2\alpha\beta \int_0^\infty u^2\,\dd r.
\ee

To get a lower estimate of the energy, we take $v$ as a test function satisfying $v(0)=1$, $v(\infty)=0$, and use the BPS method to obtain
\begin{eqnarray}
I(v)
&>& \int_{0}^{\infty} \bigg\{2(v')^{2}+
\alpha\beta v^{2}
\bigg\} {\mathrm{d}}r\nn\\ \nonumber
&=&2 \int_{0}^{\infty} \bigg\{\bigg(v'+\sqrt{\frac{\alpha\beta}2}\,v\bigg)^{2}-
\sqrt{\frac{\alpha\beta}2} (v^2)'
\bigg\} {\mathrm{d}}r\\ 
&\geq&
\sqrt{2\alpha\beta},\label{46}
\end{eqnarray}
so that the lower bound is attained when $v$ solves $v'+\sqrt{\frac{\alpha\beta}2}\,v=0$ or
 \be\label{47}
v=\e^{-\sqrt{\frac{\alpha\beta}2}\,r},
\ee
which happens to be the lower bound function in \eqref{x64}.
On the other hand, using \eqref{47} as a test function, we obtain an upper estimate for the energy
\begin{eqnarray}\label{48}
I(u)\leq I(v)
&=& \int_{0}^{\infty} \bigg\{2\left(v'+\sqrt{\frac{\alpha\beta}2}v\right)^{2}-
\sqrt{2\alpha\beta} (v^2)'+\frac{(1-v^{2})^{2}}{r^{2}}
\bigg\}{\mathrm{d}}r\nn \\ \nonumber
&=& \sqrt{2\alpha\beta}+\int_{0}^{\infty}\frac{(1-\e^{-\sqrt{2\alpha\beta}\,r})^{2}}{r^{2}}{\mathrm{d}}r\nn\\
&=&\sqrt{2\alpha\beta}\left(1+\int_0^\infty \frac{(1-\e^{-r})^{2}}{r^{2}}{\mathrm{d}}r\right)=\sqrt{2\alpha\beta}(1+2\ln2).
\end{eqnarray}
Therefore, combining \eqref{46} and \eqref{48}, we get the lower and upper estimates of the energy $I(u)$ as follows:
\begin{eqnarray}\label{59}
\sqrt{2\alpha\beta}< I(u)<\sqrt{2\alpha\beta}(1+2\ln2),
\end{eqnarray}
where the right-hand side inequality in \eqref{59} is strict since \eqref{47} does not satisfy \eqref{42}.

We note that the profile \eqref{47} suggests that we may further improve the upper bound in \eqref{59} by taking a trial 
{\em undetermined} profile $v_a(r)
=\e^{-a r}$ ($a>0$) and minimizing the function 
\bea
F(a)&=&I(v_a)=a+\frac{\alpha\beta}{2a}+\int_0^\infty \frac{(1-\e^{-2a r})^2}{r^2}\,\dd r\nn\\
&=&a\left(1+4\ln 2\right)+\frac{\alpha\beta}{2a},
\eea
giving rise to the solution
\be
a_0^2=\frac{\alpha\beta}{2(1+4\ln 2)},\quad F(a_0)=\sqrt{2\alpha\beta(1+4\ln2)}.
\ee
Consequently we obtain an improvement upon \eqref{59}:
\be\label{612}
\sqrt{2\alpha\beta}< I(u)<\sqrt{2\alpha\beta(1+4\ln2)},
\ee
where the right-hand side inequality is again strict because $v_a(r)$ does not satisfy \eqref{42}.
This result is \eqref{es} which presents a significant improvement over \eqref{59} since $(1+2\ln 2)-\sqrt{(1+4\ln2)}>\frac7{16}$.

It will be interesting to study whether the boundary value problem \eqref{42} has a unique solution without assuming $u\geq0$.
\medskip

Our study may find applications in other related problems.
As an example, consider the $SO(3)$ Georgi--Glashow model \cite{KZ,W} described by the Lagrangian density
\be
{\cal L}=-\frac14G^a_{\mu\nu}G^{a\mu\nu}-\frac12 (D_\mu\phi)^a (D^\mu\phi)^a -\frac\lm8\left(\phi^a\phi^a-\frac{m^2}\lm\right)^2,
\ee
where $a=1,2,3$ is the group index, $A_\mu=(A^a_\mu)$ a gauge field, $\phi$ a scalar field in the adjoint representation of the gauge group, $m,\lm>0$, and
\be
G_{\mu\nu}=\pa_\mu A_\nu -\pa_\nu A_\mu +e[A_\mu,A_\nu],\quad D_\mu \phi=\pa_\mu\phi+e[A_\mu,\phi],
\ee
are the field strength tensor and  gauge-covariant derivative, such that the spontaneously broken symmetry results in the vector
field mass $M_W$, Higgs boson mass $M_H$, and the mass ratio $\epsilon$, given by \cite{KZ}
\be
M_W=\frac{em}{\sqrt{\lm}},\quad M_H=m,\quad \epsilon=\frac{M_H}{M_W}=\frac{\sqrt{\lm}}e,
\ee
 respectively.
The static spherically symmetric monopole soliton assumes the hedgehog form \cite{KZ,Pol,tH}:
\be
\phi^a=\frac{H(r)}{er^2} x^a,\quad A^a_0=0,\quad A^a_i=\epsilon_{aij} x^j\frac{(1-K(r))}{er^2},
\ee
whose energy is  $E=-\int_{\bfR^3}{\cal L}\,\dd x$ which in turn is given by \cite{KZ}
\be\label{75}
E=\frac1{c_0}\int_0^\infty\left\{(K')^2+\frac{(K^2-1)^2}{2r^2}+\frac{H^2 K^2}{r^2}+\frac{(rH'-H)^2}{2r^2}
+\frac{\lm r^2}{8e^2}\left(\frac{H^2}{r^2}-\frac{m^2 e^2}{\lm}\right)^2\right\}\,\dd r,
\ee
where $c_0=\frac{e^2}{4\pi}$ is the fine-structure constant. Furthermore, with the rescaled radial variable $M_W r\mapsto r$
and the substitution $u=K,f=\frac Hr$, the
energy \eqref{75} becomes
\be
E=\frac{M_W}{c_0}\, C(\epsilon),
\ee
where \cite{FOR,KZ}
\be\label{77}
C(\epsilon)=\int_0^\infty\left\{(u')^2+\frac{(u^2-1)^2}{2r^2}+{f^2 u^2}+\frac{r^2}2(f')^2
+\frac{\epsilon^2}{8}r^2\left(f^2-1\right)^2\right\}\,\dd r,
\ee
and $u,f$ are subject to the same boundary conditions stated in \eqref{1.2}. This functional is covered as a special case of \eqref{111}. In particular, $C(0)=1$ since the minimizer of the right-hand side of \eqref{77} is given by the BPS solution \eqref{317}
with $\alpha=1$, and $C(\infty)$ is given by setting $f=1$ in \eqref{77} with $u$ a minimizer of the resulting reduced functional. That is,
\be
C(\infty)=\min\left\{\int_0^\infty\left\{(u')^2+\frac{(u^2-1)^2}{2r^2}+{u^2}\right\}\,\dd r\,\bigg|\, u(0)=1,u(\infty)=0\right\}.
\ee
Hence, in view of the result \eqref{es} with $\alpha\beta=2$, we arrive at the estimates
\be
1< C(\infty)< \sqrt{1+4\ln 2}\approx 1.9423153.
\ee

Note that, in \cite{FOR,KZ},  by using numerical solutions, it is estimated  that $C(\infty)=1.787$. This result is consistent with our estimates above.

\medskip

Our methods may also be applied to other  monopole and dyon existence problems
 including those formulated in \cite{Burzlaff,COFN,DT,KM,LY,LW,Wein}.
\medskip

The authors would like to thank an anonymous referee whose thoughtful suggestions helped improve the manuscript.

\end{document}